\documentclass{article}
\usepackage{dcolumn}
\usepackage{bm}
\usepackage{verbatim}       

\usepackage[dvips]{graphicx}
\usepackage{amsthm}
\usepackage{amssymb}
\usepackage{bm}
\usepackage{amstext}
\usepackage{psfrag}
\usepackage{amsmath}
\usepackage{amsfonts}

\newfont{\bb}{msbm10 at 12pt}

\newcommand{\p}{\partial}

\newcommand{\dd}{{\rm d}}
\newcommand{\bd}{\begin{definition}}                
\newcommand{\ed}{\end{definition}}                  
\newcommand{\bc}{\begin{corollary}}                 
\newcommand{\ec}{\end{corollary}}                   
\newcommand{\bl}{\begin{lemma}}                     
\newcommand{\el}{\end{lemma}}                       
\newcommand{\bp}{\begin{proposition}}            
\newcommand{\ep}{\end{proposition}}                
\newcommand{\bere}{\begin{remark}}                  
\newcommand{\ere}{\end{remark}}                     

\newcommand{\bt}{\begin{theorem}}
\newcommand{\et}{\end{theorem}}

\newcommand{\be}{\begin{equation}}
\newcommand{\ee}{\end{equation}}

\newcommand{\bit}{\begin{itemize}}
\newcommand{\eit}{\end{itemize}}
\newtheorem{theorem}{Theorem}[section]
\newtheorem{corollary}[theorem]{Corollary}

\newtheorem{lemma}[theorem]{Lemma}
\newtheorem{proposition}[theorem]{Proposition}
\theoremstyle{definition}
\newtheorem{definition}[theorem]{Definition}
\theoremstyle{remark}
\newtheorem{remark}[theorem]{Remark}

\begin{document}
%

\title{ Weak distinction and the optimal definition of causal continuity\thanks{Published at: Class. Quantum Grav. 25 (2008) 075015.
}}

\author{E. Minguzzi\footnote{Department of Applied Mathematics, Florence
 University, Via S. Marta 3,  50139 Florence, Italy. E-mail: ettore.minguzzi@unifi.it}}

\date{}
\maketitle

\begin{abstract}
\noindent Causal continuity is usually defined by imposing the
conditions (i) distinction and (ii) reflectivity. It is proved here
that a new causality property which stays between weak distinction
and causality, called {\em feeble distinction}, can actually replace
distinction in the definition of causal continuity.  An intermediate
proof shows that feeble distinction and future (past) reflectivity
implies past (resp. future) distinction. Some new characterizations
of weak distinction and reflectivity are given.
\end{abstract}

%


\section{Introduction}

Recently Bernal and S\'anchez \cite{bernal06b}  proved that causal
simplicity, usually defined by imposing the two properties \cite[p.
65]{beem96}, (a) distinction   and (b) for all $x \in M$ the sets
$J^{+}(x)$ and $J^{-}(x)$  are closed (this property is equivalent
to $J^+=\bar{J}^+$, see \cite[sect. 3.10]{minguzzi06c}), can
actually be improved by replacing (a) with the weaker requirement of
causality. In this work I give a result which goes in the same
direction of optimizing the definitions and results underlying the
causal hierarchy  of the spacetimes.

Causal continuity is usually defined by imposing the conditions
\cite[p. 294]{hawking74} \cite[p. 59,70]{beem96} (i) distinction and
(ii) reflectivity. The distinction condition was defined, quite
naturally, by Hawking and Sachs \cite[p. 292]{hawking74} as the
imposition of both future and past distinction. At the time,
Kronheimer and Penrose had already defined the past, future and the
weak distinction properties \cite[p. 486]{kronheimer67} as follows:
a spacetime is future distinguishing if $I^{+}(x)=I^+(z) \Rightarrow
x=z$; past distinguishing if $I^{-}(x)=I^-(z) \Rightarrow x=z$; and
weakly distinguishing if ``$I^{+}(x)=I^+(z) \textrm{ and }
I^{-}(x)=I^-(z)$'' $\Rightarrow x=z$.

Clearly, future (past) distinction implies weak distinction and
there are examples of spacetimes which are weakly distinguishing but
neither future nor past distinguishing (see figure \ref{wdis}(B)),
thus weak distinction is a strictly weaker property than future or
past distinction. Nevertheless, in this work I am going to prove
(corollary \ref{cor}) that condition (i) defining causal continuity
can be replaced with (i') {\em feeble distinction}, a property
which, as I will show, is even weaker than weak distinction. This
result comes from a interesting lemma which mixes future and past
properties (otherwise usually found separated in other theorems),
namely feeble distinction and future (past) reflectivity implies
past (resp. future) distinction (theorem \ref{psr}).

I denote with $(M,g)$ a $C^{r}$ spacetime (connected, time-oriented
Lorentzian manifold), $r\in \{3, \dots, \infty\}$ of arbitrary
dimension $n\geq 2$ and signature $(-,+,\dots,+)$. On $M\times M$
the usual product topology is defined. The subset symbol $\subset$
is reflexive, i.e. $X \subset X$. The closure of the causal future
on $M\times M$ is denoted $A^+$, that is, $A^{+}=\bar{J}^+$. For
other notations concerning causal sets the reader is referred to
\cite{minguzzi07b}.

\begin{figure}
\centering \psfrag{R}{Remove segment} \psfrag{P}{Remove point}
\psfrag{A}{\!\!\!\!(A)}  \psfrag{B}{\!\!\!\!(B)}
\includegraphics[width=8cm]{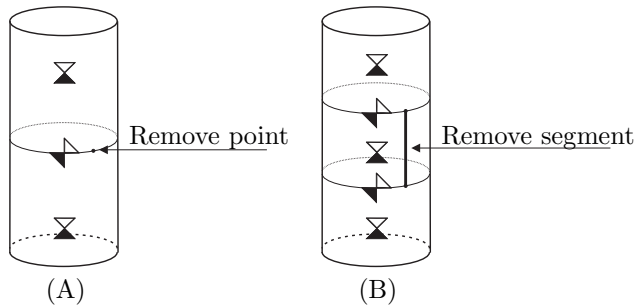}
\caption{(A). A non-total imprisoning (hence causal), reflecting,
non-feebly distinguishing (and hence non-causally continuous)
spacetime. (B). A weakly distinguishing spacetime which is neither
future nor past distinguishing. The causal relation $D^{+}$ is
antisymmetric while $D^{+}_p$ and $D^{+}_f$ are not. Note that both
future and past reflectivity fail to hold as one should expect from
theorem \ref{psr}. } \label{wdis}
\end{figure}

\section{Weak distinction}

Since the property of weak distinction has been only marginally used
in causality theory I devote a few pages to its study, in particular
I develop some equivalent characterizations. The reader is assumed
to be familiar with the approach to causal relations as subsets of
$M\times M$ (see \cite{minguzzi06c} and \cite{minguzzi07b}).

Recall  that the relations on $M$
\begin{align*}
D^{+}_f&=\{(x,y): y \in \overline{I^{+}(x)}\,\},\\
D^{+}_p&=\{(x,y): x \in \overline{I^{-}(y)}\,\},
\end{align*}
are reflexive and transitive. Moreover, the spacetime is future
(past) distinguishing iff $D^{+}_f$ (resp. $D^+_p$) is antisymmetric
\cite{minguzzi07b}. Define $D^+=D^+_p\cap D^+_f$ so that
\begin{equation} \label{nis}
D^+=\{(x,y): y \in \overline{I^{+}(x)}\textrm{ and } x \in
\overline{I^{-}(y)} \}.
\end{equation}
Recall \cite{minguzzi07b} that $D^{+}_f=A^+$ iff the spacetime is
future reflecting and $D^{+}_p=A^+$ iff the spacetime is past
reflecting (for other equivalent characterizations of reflectivity
see \cite{minguzzi06c}). Since $D^+_p,D^+_f \subset A^+$, it is
$D^+=A^+$ iff $D^{+}_f=A^+$ and $D^+_p=A^+$. This observation reads
\begin{lemma}
A spacetime $(M,g)$ is reflecting iff $D^{+}=A^{+}$.
\end{lemma}

The antisymmetry condition for $D^{+}$ is equivalent to weak
distinction, indeed it holds.

\begin{lemma} \label{mos}
The following conditions on a spacetime $(M,g)$ are equivalent
\begin{enumerate}
\item\label{1}  $I^{+}(x)= I^{+}(y)$ and $I^{-}(x)= I^{-}(y)$ imply
$x=y$.
\item \label{2} $D^{+}$ is antisymmetric.
\item \label{4} The map $M \to P(M)$ defined by $x \to I^{+}(x) \cup
I^{-}(x)$ is injective.
\item \label{5} The map $M \to P(M)$ defined by $x \to D^{+}(x)$ is injective.
\item \label{6} The map $M \to P(M)$ defined by $x \to D^{-}(x)$ is injective.
\end{enumerate}
\end{lemma} \ref{1}
 $\Rightarrow$ \ref{2}. Assume $D^{+}$ is not antisymmetric then
there are $x$ and $y$, $x\ne y$ such that $(x,y) \in D^{+}$ and
$(y,x) \in D^{+}$ which reads ``$y \in \overline{I^{+}(x)}$ and $x
\in \overline{I^{-}(y)}$ and $x \in \overline{I^{+}(y)}$ and $y \in
\overline{I^{-}(x)}$''. $y \in \overline{I^{+}(x)}$ and $x \in
\overline{I^{+}(y)}$ implies $I^{+}(y)=I^{+}(x)$ while $x \in
\overline{I^{-}(y)}$ and $y \in \overline{I^{-}(x)}$ implies
$I^{-}(y)=I^-(x)$, thus \ref{1} does not hold.

\ref{2} $\Rightarrow$ \ref{1}. Assume \ref{1} does not hold. There
are $x \ne y$ such that $I^{+}(x)= I^{+}(y)$ and $I^{-}(x)=
I^{-}(y)$, thus $y \in \overline{I^{+}(y)}=\overline{I^{+}(x)}$,  $x
\in \overline{I^{-}(x)}=\overline{I^{-}(y)}$, and analogously with
the roles of $x$ and $y$ exchanged. Thus $(x,y) \in D^{+}$ and
$(y,x) \in D^{+}$, i.e. $D^+$ is not antisymmetric.

%

\ref{4} $\Rightarrow$ \ref{1}. Indeed if \ref{1} does not hold there
are $x\ne y$, such that $I^{+}(x)=I^{+}(y)$ and $I^{-}(x) =
I^{-}(y)$, thus $I^{+}(x)\cup I^{-}(x)=I^{+}(y)\cup I^{-}(y)$, and
hence \ref{4} does not hold, a contradiction.

\ref{1} $\Rightarrow$ \ref{4}. First, \ref{1} implies that $(M,g)$
is chronological indeed, the existence of $w\ne z$ with $w \ll z \ll
w$ would imply $I^{+}(w)=I^{+}(z)$ and $I^{-}(w) = I^{-}(z)$ which
contradicts 1. Since $(M,g)$ is chronological for every event $z$
the sets $I^{+}(z)$ and $I^{-}(z)$ are disjoint. Assume \ref{4} does
not hold then there are $x\ne y$, such that $I^{+}(x)\cup
I^{-}(x)=I^{+}(y)\cup I^{-}(y)$, but given $x' \gg x$ it is $x' \ll
y$ or $x' \gg y$. But the former possibility can not hold because it
implies  $x \ll y$  thus $y$ must belong to $I^{+}(y)$ or $I^{-}(y)$
both cases implying a violation of chronology. Thus $I^{+}(x)
\subset I^{+}(y)$, and changing the roles of $x$ and $y$, $I^{+}(y)
\subset I^{+}(x)$, i.e. $I^{+}(x) = I^{+}(y)$. Changing the roles of
past and future $I^{-}(x) = I^{-}(y)$ which contradicts \ref{1},
hence, by contradiction, \ref{4} must hold.

\ref{2} $\Leftrightarrow$ (\ref{5} and \ref{6}). It follows from
theorem 2.3(c) of \cite{minguzzi07b}.

\begin{definition}
A spacetime is weakly distinguishing if it satisfies the equivalent
properties of lemma \ref{mos}.
\end{definition}

Note that $D^+$ is transitive and reflexive because $D^+_f$ and
$D^+_p$ are transitive and reflexive, moreover it is trivial to
prove that if $D^+_f$ or $D^+_p$ is antisymmetric then $D^+$ is
antisymmetric. This observation gives another proof of the well
known, already mentioned in the introduction, result that

\begin{lemma}
If $(M,g)$ is past or future distinguishing then it is weakly
distinguishing.
\end{lemma}

\begin{remark}
It is easy to prove that past or future distinction at $x$ implies
weak distinction at $x$, however, it is a non trivial matter to find
a counterexample of the converse. An example of a weakly
distinguishing spacetime in which there is a event $p$ at which the
spacetime is neither future nor past distinguishing can be obtained
from the spacetime of figure \ref{feeble} by removing the point $q$.
\end{remark}

\begin{remark}
The possibility of expressing weak distinction through the
injectivity of the map $x \to D^+(x)$,  is a consequence of the
reflexivity and transitivity  of $D^+$ (see theorem 2.3(c) of
\cite{minguzzi07b}). Actually, other causality properties such as
strong causality can be characterized in terms of the injectivity of
a suitable causal set function although, as far as I know, strong
causality cannot be  obtained as an antisymmetry condition for a
suitable reflexive and transitive causal relation. The idea of
expressing the causality conditions as injectivity conditions on
causal set functions goes back to I. R\'acz \cite{racz87}.
\end{remark}

Weak distinction and the  causal structures of Kronheimer and
Penrose \cite{kronheimer67} are intimately related as the next two
lemmas prove.

\begin{lemma}
$I^{+}$ is a left $D^{+}_p$-ideal, that is, $I^{+}\subset D^+_p$ and
$ D^{+}_p \circ I^{+} \subset I^{+}$, analogously $I^{+}$ is a right
$D^{+}_f$-ideal, that is, $I^{+}\subset D^+_f$ and $  I^{+} \circ
D^{+}_f \subset I^{+}$.

$I^{+}$ is a $D^{+}$-ideal, and the triple $(M,D^{+},I^{+})$ is a
causal structure in the sense of Kronheimer and Penrose iff the
spacetime is weakly distinguishing.
\end{lemma}

\begin{proof}
It follows trivially from the definitions and from the fact that
$I^{+}$ is open.
\end{proof}

Let $\{R^+_\alpha\}$ be a set of relations labeled by an index
$\alpha$. Note that if $R_{\alpha}^+ \circ I^+ \subset I^+$ for
every $\alpha$ then $\bigcup_\alpha R_\alpha^+ \circ I^{+} \subset
I^{+}$, thus there is the largest set with respect to which $I^{+}$
is a left ideal. This largest set necessarily contains $I^{+}$
because $I^+\circ I^+\subset I^+$. Analogous considerations hold for
the property $I^+\circ R^{+} \subset I^+$ and `` $I^+\circ R^{+}
\subset I^+$ and $R^+ \circ I^+ \subset I^+$''.

\begin{lemma}
The set $D^{+}_p$ is the largest set which satisfies $D^+_p \circ
I^+ \subset I^+$, analogously,   $D^{+}_f$ is the largest set which
satisfies $ I^+ \circ D_f^{+} \subset I^+$ and $D^{+}$ is the
largest set  which satisfies both $D^{+} \circ I^+ \subset I^+$ and
$ I^+ \circ D^{+} \subset I^+$. In particular if $(M, R^+, I^{+})$
is a causal structure in the sense of Kronheimer and Penrose then
$R^{+} \subset D^+$.
\end{lemma}

\begin{proof}
We already know that $ D^{+}_p \circ I^{+} \subset I^{+}$. Assume
there is $R^{+} \supset D^{+}_p$, $R^{+}\ne D^{+}_p$, such that $
R^{+} \circ I^{+} \subset I^{+}$. Take $(x,z) \in R^{+}\backslash
D^{+}_p$ then for every $y$ such that $(z,y) \in I^{+}$ it is $(x,y)
\in I^{+}$, but $D^{+}_p$ can be characterized as $D^{+}_p=\{(x,z)
\in M: \forall y \in I^{+}(z) \textrm{ it is } x \in I^{-}(y)\}$
(see the proof of lemma 4.2 \cite{minguzzi07b}), thus $(x,z) \in
D^{+}_p$ a contradiction. Analogously, $D^{+}_f$ is the largest set
 which satisfies $ I^+ \circ D^{+}_f \subset I^+$.

 We already know that $D^{+} \circ I^+ \subset I^+$ and $ I^+ \circ
D^{+} \subset I^+$. If $D^{+}=D^{+}_f\cap D^{+}_p$ is not the
largest set which satisfies this property then there is $(x,z)
\notin D^+$,  such that $(x,z) \circ I^{+} \subset I^{+}$ and $ I^+
\circ (x,z) \subset I^+$. The pair $(x,z)$ can't belong to both
$D^+_p$ and $D^+_f$, assume without loss of generality $(x,z)\notin
D^{+}_p$, then $R^{+}=D^{+}_p \cup (x,z)$ is larger than $D^{+}_p$
and satisfies $R^{+} \circ I^+\subset I^+$ a contradiction.

\end{proof}

\section{Feeble distinction}

Future distinction is equivalent to the antisymmetry of $D^+_f$
however, it is also equivalent \cite{minguzzi07b}  to an apparently
weaker requirement, namely $(x,z) \in J^{+}$ and $(z,x) \in D^{+}_f
\Rightarrow x=z$. It is natural to ask whether weak distinction can
be expressed as: $(x,z) \in J^{+}$ and $(z,x) \in D^{+} \Rightarrow
x=z$. As we shall see, the answer is negative and the property
defines a new level in the causal ladder which stays between weak
distinction and causality.

\begin{lemma} \label{lhf}
The following conditions on a spacetime $(M,g)$ are equivalent
\begin{enumerate}
\item $(x, y) \in J^{+}$, $I^{+}(x)= I^{+}(y)$ and $I^{-}(x)= I^{-}(y)$ imply
$x=y$.
\item $(x,y) \in J^{+}$ and $(y,x) \in D^{+} \Rightarrow x=y$. \\
\end{enumerate}

\end{lemma}

\begin{proof}
1 $\Rightarrow$ 2. Assume 2 does not hold then there are $x$ and
$y$, $x < y$ such that $(y,x) \in D^{+}$ which reads ``$x<y$ and $x
\in \overline{I^{+}(y)}$ and $y \in \overline{I^{-}(x)}$''. $y \in
\overline{I^{+}(x)}$ and $x \in \overline{I^{+}(y)}$ implies
$I^{+}(y)=I^{+}(x)$ while $x \in \overline{I^{-}(y)}$ and $y \in
\overline{I^{-}(x)}$ implies $I^{-}(y)=I^-(x)$, thus \ref{1} does
not hold.

2 $\Rightarrow$ 1. Assume 1 does not hold. There are $x < y$ such
that $I^{+}(x)= I^{+}(y)$ and $I^{-}(x)= I^{-}(y)$, thus $x \in
\overline{I^{+}(x)}=\overline{I^{+}(y)}$,  $y \in
\overline{I^{-}(y)}=\overline{I^{-}(x)}$. Thus $(x,y) \in J^{+}$ and
$(y,x) \in D^{+}$, but $x \ne y$, a contradiction.
\end{proof}

\begin{definition}
A spacetime is said to be {\em  feebly  distinguishing} if it
satisfies one of the equivalent properties of lemma \ref{lhf}. In
short, a spacetime is feebly distinguishing if there is no pair of
causally related events with the same chronological pasts and
futures.
\end{definition}

\begin{lemma}
If a spacetime is weakly distinguishing then it is feebly
distinguishing. If a spacetime if feebly distinguishing then it is
causal.
\end{lemma}

\begin{proof}
It follows trivially from the fact that $J^{+}\subset D^{+}$.
\end{proof}

It is easy to check that feeble distinction differs from causality,
see for instance the spacetime  of figure \ref{wdis}(A). It is
instead a non trivial matter to establish that feeble distinction
differs from weak distinction.

Figure \ref{feeble} gives an example of feebly distinguishing
non-weakly distinguishing spacetime.

\begin{figure}
\centering \psfrag{x}{$x$} \psfrag{y}{$y$} \psfrag{t}{$t$}
\psfrag{p}{$p$} \psfrag{q}{$q$} \psfrag{r}{$r$} \psfrag{s}{$s$}
\psfrag{k}{$\!\!K$} \psfrag{c}{$\Gamma^+_f(p)$}
\psfrag{d}{$\Gamma^+_f(p')$} \psfrag{f}{$F$} \psfrag{w}{$p'$}
\psfrag{indentify}{Identify} \psfrag{remove}{remove $K'$}
\psfrag{removed}{remove filter}
\includegraphics[width=10cm]{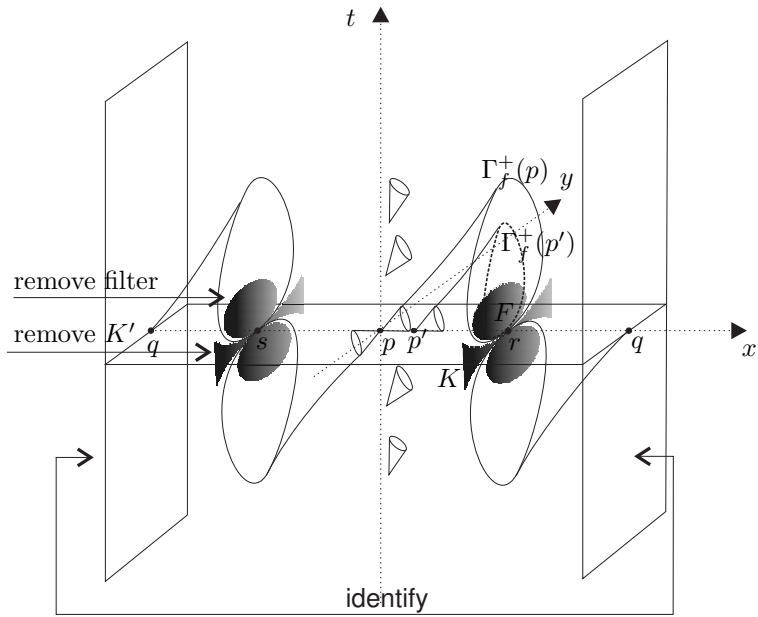}
\caption{An example of feebly distinguishing non-weakly
distinguishing spacetime. The shaded gray regions have been removed
from the manifold. The events $p$ and $q$ share the same
chronological past and future and form the only pair of events with
this property. They are not causally related thanks to the removed
sets, thus the spacetime is feebly distinguishing.} \label{feeble}
\end{figure}

\subsection{The spacetime example}
Since the construction of the spacetime is quite involved I offer
the next explanation.

The starting point is a spacetime $\mathbb{R}\times S^1 \times
\mathbb{R}$ of coordinates $(t,x,y)$, $x \in [-2, 2]$, in which the
metric has Killing vectors $\p/\p y$, $\p/\p x$, and the qualitative
behavior of the light cones is as displayed in  figure \ref{feeble}.
The analytical expression of the metric is not important, but the
expression $\dd s^2=-(\cosh t-1)^2(\dd t^2-\dd x^2)-\dd t \dd x+ \dd
y^2$ would be fine. Consider the points $p=(0,0,0)$, $r=(0,1,0)$,
$s=(0,-1,0)$ and $q=(0,2,0)=(0,-2,0)$.

I am going to remove suitable sets from the surfaces $x=\pm1$, so as
to make the spacetime weakly distinguishing but for the events
$p,q$, which actually will share the same chronological past and
future. Since $p$ and $q$ will not be causally related the spacetime
so obtained, denoted $(M,g)$, will be feebly distinguishing but
non-weakly distinguishing.

The closed sets to be removed are as follows. First a set $K$, on
the  surface $x=1$,  should be removed in order to guarantee that,
if weak distinction is violated, then the violation  must happen on
the $x$ axis. The set $K$ has two boundaries which are tangent to
the null plane $t=0$, and whose radius of curvature at the point of
contact with the plane must be sufficiently large so that the causal
futures and pasts of the points lying in the $x$ axis can intersect
it only after one complete loop (i.e. the timelike  curves issuing
from a point on the $x$ axis can intersect $K$ only from the second
intersection with the plane $x=+1$). The figure displays that
another set $K'$ of the same shape has been removed, but this is
done for symmetry purposes, and is not strictly necessary.

Let $z$ be a point in the $x$ axis. Denote with $\Gamma_f^+(z)$
 the points of first intersection  of the
causal curves issuing from $z$ with the surface $x=+1$. This set has
a boundary whose equation $t(y)=h^+_f(z)(y)$ has a Taylor expansion
at $y=0$ which can be easily determined with a little algebra. An
analogous set $\Gamma_f^-(z)$ (where the minus sign denotes that the
intersection of the causal curves is with the surface $x=-1$) and
past versions can be defined. The point here is that the first terms
of this Taylor expansion are quadratic and the coefficient decreases
with the distance of $z$ from the surface. The idea is to remove
inside $\Gamma_f^+(p)$, $\Gamma_p^-(p)$, $\Gamma_f^-(q)$ and
$\Gamma_p^{+}(q)$ suitable sets passing through $s$ and $r$ so that
``$(p,q) \in D^{+}$ and $(q,p) \in D^{+}$"  would still hold as the
future (resp. past) of $p$ and $q$ would pass `below' (resp.
`above') the removed sets. For instance, $q$ would still remain in
the closure of the causal future of $p$. The sets to be removed will
be called `filters' and have an elliptical shape in the figure.

The filters are chosen so that $p$ and $q$ are the only two points
with the same past and future. Let us focus on $\Gamma_f^+(p)$. If
its boundary has Taylor expansion $h^+_f(p)(y)={a_2(p)} y^2
+{a_4(p)} y^4 +\ldots$, choose the set to be removed so that its
boundary has Taylor expansion $a_2(p) y^2 +2\vert{a_4(p)} y^4
+\ldots\vert$. Now, what happens is that chosen $p'$ between $p$ and
$r$, since $a_2(p')>a_2(p)$, there is a  $\epsilon(p')>0$ such that
the points in $\Gamma_f^+(p')$ with $t\le \epsilon(p')$ belong to
the removed set. In other words, the causal curves issued from $p'$
may `pass' the filter but if so they are forced to move towards
events with $t>\epsilon(p')$ and then, because of the shape of the
metric, they can't return arbitrarily close to the $x$ axis. As a
result, for instance $p \notin \overline{I^{+}(p')}$.

Note that the points on the segments $(s,p]$ and $(r,q]$  on the $x$
axis have the same chronological future, and the points on the
segments  $[p,r)$ and $[q,s)$ have the same chronological past.

Another method, perhaps simpler, to define the filter is as follows.
I describe it for the filter denoted $F$ in the figure, the other
cases being analogous. Take a sequence $p_n=(0,1/n,0)$, $p_n \to p$,
and define $U_n=\{z=(a,b,c): a\le 1/n, b=1, c \in \mathbb{R}\}$.
Finally, define $F=\bigcup_n \Gamma_f^+(p_n)\cap U_n$. This filter
has the same causal effect of the one described using the Taylor
expansion, however, note that it has a characteristic flower shape
and not an elliptic one as in the figure.

\section{Strengthening the definition of causal continuity}

Finally, I give the proof to the  results mentioned in the
introduction.

\begin{theorem} \label{psr}
If the spacetime is future (past) reflecting  then $D^{+}=D^{+}_p$
(resp. $D^{+}=D^{+}_f$). In particular, feeble distinction and
future (past) reflectivity imply past (resp. future) distinction.
\end{theorem}

\begin{proof}
If the spacetime is future reflecting then $D^{+}_f=A^{+}$, but
$D^{+}_p \subset A^{+}$ thus $D^{+}=D^+_p \cap D^+_f=D^{+}_p$. Thus
under future reflectivity, $(x,y) \in J^{+}$ and $(y,x) \in D^{+}_p$
is equivalent to $(x,y) \in J^{+}$ and $(y,x) \in D^{+}$ which  by
feeble distinction implies $x=y$.
%
%
The proof in the other case is
analogous.

\end{proof}

\begin{corollary} \label{cor}
A spacetime is causally continuous iff it is feebly distinguishing
and reflecting.
\end{corollary}

\begin{proof}
The only if part is trivial as it follows from the usual definition
of causal continuity as a spacetime which is distinguishing and
reflecting. For the if part note that by theorem \ref{psr}, feeble
distinction and past reflectivity imply future distinction.
Moreover, feeble distinction and future reflectivity imply past
distinction, thus feeble distinction and reflectivity imply
distinction. Thus the spacetime  is distinguishing and using again
the assumed reflectivity the causal continuity follows.
\end{proof}

Feeble distinction implies causality, however, in the definition of
causal continuity causality cannot replace feeble distinction,
indeed it is quite easy to construct an example of spacetime which
is causal, reflecting and non-feebly distinguishing (and hence
non-causally continuous), see figure \ref{wdis}(A). Actually, this
example is also non-total imprisoning. I shall prove in a related
work \cite{minguzzi07f} that feeble distinction implies non-total
imprisonment which implies causality. However, in the definition of
causal continuity feeble distinction can not even be relaxed to
non-total imprisoning as the example of figure \ref{wdis}(A) again
proves.

\section{Conclusions}

The property of weak distinction has been studied showing that it is
equivalent to the antisymmetry of the causal relation $D^+$ defined
by Eq. (\ref{nis}). The set $D^+$ enters also in an alternative
definition of reflectivity, namely the condition $D^+=A^+$.

Between weak distinction and causality I defined another level,
called {\em feeble distinction}. Examples have been provided  which
show that feeble distinction indeed differs from weak distinction
and causality (actually it differs from non-total imprisonment).

Next a basic step has been the proof that feeble distinction and
future (past) reflectivity implies past (resp. future) distinction,
a curious statement that mixes future and past properties. Using it,
it has been finally  shown that in the definition of causal
continuity it is possible to replace the distinction property with
feeble distinction. Some known examples prevent the possibility of
weakening the feeble distinction property to the level which stays
immediately below it in the causal ladder. Since the causal ladder
is not fixed, and new levels can always be found, there is some
natural uncertainty on what this optimality could mean. In any case
I will show in a related work \cite{minguzzi07f} that the non-total
imprisonment property stays between feeble distinction and
causality, and  figure \ref{wdis}(A) proves that in the definition
of causal continuity, feeble distinction can not be replaced by
non-total imprisonment. In this sense, the definition of causal
continuity given in this work is optimal.


\end{document}